\documentclass[letterpaper, 10 pt, conference]{ieeeconf}  
\IEEEoverridecommandlockouts                              

	\usepackage{graphics}      
\usepackage{amsmath}
\usepackage{amsfonts}
\usepackage{subcaption}
\usepackage{epsfig}
\usepackage{algorithm}
\usepackage{algpseudocode}
\usepackage{csquotes}
\usepackage{verbatim}
\usepackage{booktabs}
\usepackage{subcaption}
\usepackage{cite}
\usepackage{xcolor}

\usepackage{comment}
\usepackage{caption}
\usepackage{multirow}

\newtheorem{theorem}{Theorem}

\newtheorem{assume}{Assumption}
\newtheorem{rem}{Remark}

\newcommand{\algrulehor}[1][.2pt]{\par\vskip.5\baselineskip\hrule height #1\par\vskip.5\baselineskip}

	\title{\LARGE \bf An Improved Data Augmentation Scheme for Model Predictive Control Policy Approximation
	}

\author{Dinesh Krishnamoorthy
	\thanks{The author is with 
	the Department of Mechanical Engineering, Eindhoven University of Technology, 5600 MB, Eindhoven, The Netherlands.
		{\tt\small d.krishnamoorthy@tue.nl}}%
}

\usepackage{times}

\newif \ifpreprint
\preprinttrue

\ifpreprint
\newcommand{\solution}[1]{#1}
\else
\newcommand{\solution}[1]{}
\fi

\begin{document}

\maketitle
		\thispagestyle{empty}
\pagestyle{empty}

\begin{abstract}%
	This paper considers the problem of data generation for MPC policy approximation. Learning an approximate MPC policy from expert demonstrations requires a  large data set consisting of optimal state-action pairs, sampled across the feasible state space. Yet, the key challenge of efficiently generating  the training samples has not been studied widely. Recently,  a  sensitivity-based data augmentation  framework  for MPC policy approximation was proposed, where the parametric sensitivities are exploited to cheaply generate several additional   samples from a single offline MPC computation. The error due to augmenting the training data set with inexact samples was shown to increase with the 
	size of the neighborhood around each sample used for data augmentation. Building upon this work, this letter paper presents an improved data augmentation scheme based on predictor-corrector steps that enforces a user-defined level of accuracy, and shows that the error bound of the augmented samples are independent of the size of the neighborhood used for data augmentation. 
	
\end{abstract}


\section{Background}
\subsection{Pre-computed control policies}
Implementing optimization-based controllers such as  model predictive control (MPC) for fast dynamic systems  with limited computing and memory capacity has motivated research on pre-computing and storing the optimal control policy offline such that it can be used online without recomputing the optimization problem. This was first studied under the framework of explicit MPC for linear time-invariant systems \cite{bemporad2002explicit}. However, the synthesis of the explicit MPC law does not scale well, since the number of piecewise affine polyhedral regions required to capture the MPC control law can increase exponentially with problem size. Extension to nonlinear systems and economic objectives are also not trivial. 

An alternative approach is to approximate the optimal policy using parametric function approximators, such as deep neural networks, which are broadly studied under the context of \emph{learning from demonstrations}. The idea of approximating an MPC policy using neural networks was first proposed in  \cite{parisini1995receding}, but this idea remained more or less dormant (due to the high cost of offline training).  However, with the recent promises of  deep learning, there has been an unprecedented surge of interest in approximating control policies, \ifpreprint see e.g. \cite{chen2018approximating,hertneck2018learning,zhang2019safe,drgovna2018approximate, karg2018efficient,paulson2020approximate,sanchez2016learning,kumar2021industrial,aakesson2006neural,cao2020deep}. \fi
\ifpreprint This has also inspired  a number of applications ranging from energy\cite{drgovna2018approximate}, automotive \cite{zhang2020near,quan2019approximate}, chemical processing \cite{kumar2021industrial,karg2021approximateMHE}, robotics \cite{nubert2020safe}, spacecraft \cite{sanchez2016learning}, and  healthcare \cite{bonzanini2020toward} to name a few. \fi
 Research developments in this direction have been predominantly devoted to understanding the safety and performance of approximate policies \cite{hertneck2018learning,zhang2020near,paulson2020approximate,karg2021probabilistic}. Although, these are very important developments in the direction of MPC policy approximation, a major challenge of this approach that hinders practical implementation  is the cost of training the policy, which is not well studied in the literature, as also noted in \cite{chen2022large}.

\ifpreprint
\begin{figure*}
	\centering
	\includegraphics[width=0.77\linewidth]{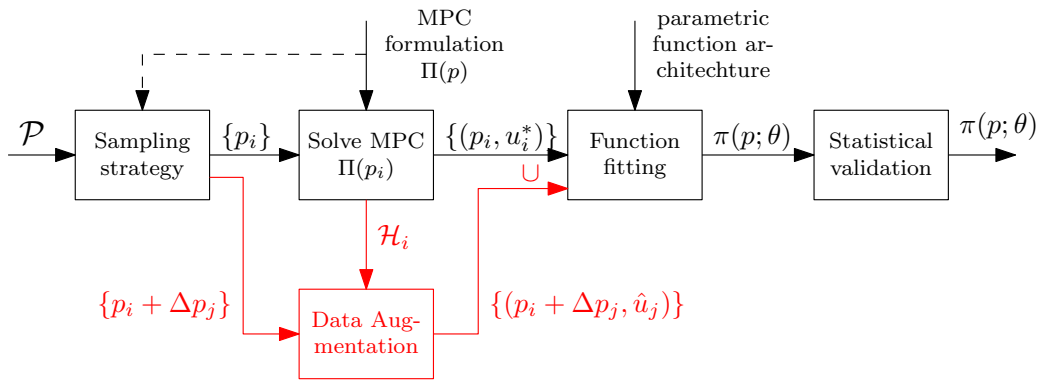}
	\caption{MPC policy approximation pipeline with  data augmentation scheme for efficient data generation.}\label{Fig:Pipeline}
\end{figure*}
\fi

\subsection{High cost of offline  learning}
Consider the \enquote{expert} policy which is given by solving  the MPC problem  
\begin{subequations}\label{Eq:MPC}
	\begin{align}
		V^*(x(t)) = \min_{x_{k},u_{k} } &\; \sum_{k=0}^{N-1} \ell_{Q}(x_{k},u_{k}) + \ell_P(x_{N})\label{Eq:MPC:cost}  \\
		\textup{s.t.} \; & x_{k+1} = f(x_{k},u_{k},d) \quad\forall k \in \mathbb{I}_{0:N-1}\label{Eq:MPC:system}\\
		& \textcolor[rgb]{0,0,0}{ x \in \mathcal{X} , \quad u \in \mathcal{U} }, \quad x_{N} \in \mathcal{X}_{f} \label{Eq:MPC:terminal}\\
		& x_{0} = x(t) \label{Eq:MPC:init}
	\end{align}
\end{subequations}
where $ x \in {X}  \subseteq \mathbb{R}^{n_{x}}$, $ u \in \mathcal{U}  \subseteq \mathbb{R}^{n_{u}}$, and $ d \in \mathbb{R}^{n_{d}} $ denote the states, inputs, and the model parameters, respectively. $f:\mathbb{R}^{n_{x}} \times \mathbb{R}^{n_{u}} \times  \mathbb{R}^{n_{d}} \rightarrow \mathbb{R}^{n_{x}} $ denotes the system model,  $ \ell_{Q}: \mathbb{R}^{n_{x}} \times \mathbb{R}^{n_{u}} \rightarrow \mathbb{R}$ 
and  $ \ell_P: \mathbb{R}^{n_{x}} \rightarrow \mathbb{R}$ denote the stage and terminal cost, \textcolor[rgb]{0,0,0}{which are parameterized by the tuning weights $ Q $ and $ P $, respectively, which are treated very generically in this letter to include different MPC formulations, including economic stage and terminal costs}.  $ N $ is the length of the prediction horizon, \textcolor[rgb]{0,0,0}{ \eqref{Eq:MPC:terminal} denotes the path and terminal constraints}, and \eqref{Eq:MPC:init} denotes the initial condition constraint. Solving \eqref{Eq:MPC} at each time step and implementing the first control input gives the implicit MPC policy  $ u^*=\pi^*(x(t)) $.

Approximating the expert policy $ \pi^*(x)  $ using supervised learning requires a data set consisting of the expert demonstrations in the form of optimal state-action pairs  $ \mathcal{D} := \{(x_{i},u_{i}^*)\}_{i=1}^{N_{s}} $, which tells us what action $ u_{i}^* $ the expert (i.e., MPC) would take at  state $ x_{i} $. Using this data set, we can learn a parameterized policy  $ \pi(x;\theta) $ that tells what  actions to take as a function of the current state $ x $, such that $ \pi(x;\theta) $ mimics $ \pi^*(x) $. The data set $ \mathcal{D} $ is generated by solving the MPC problem offline for different realizations of $ x_{i} $, ideally covering the entire \textcolor[rgb]{0,0,0}{feasible state space $\mathcal{X}_{feas}$}. 

The availability of a sufficiently rich training data set
covering the entire  state space is a key stipulation for satisfactory learning \cite{chen2022large,hertneck2018learning}. 
In fact, it has also been shown that the learned policy can  result in instability when insufficient demonstrations are used for policy fitting \cite{ross2011reduction,palan2020fitting}.
This means the MPC problem \eqref{Eq:MPC} must be solved for several different state realizations $ x_{i} $, the size of which \textcolor[rgb]{0,0,0}{increases} exponentially as the problem dimension increases, (much like the scalability issues with explicit MPC, albeit offline). For example, if we want to include at least $ s $ samples along each dimension, then we need to generate at least $ N_{s} = s^{n_{x}} $ samples. 
To this end, the first challenge is the cost of generating the data set $ \mathcal{D} := \{(x_{i},u_{i}^*)\}_{i=1}^{N_{s}} $. 

Assume now that a sufficiently rich data set $ \mathcal{D} $ has been generated offline covering $ \mathcal{X} $, and  an approximate policy $ \pi_{}(x;\theta) $ is learned and has  passed the necessary statistical validations and is certified for online use.  Now, if any of the MPC tuning parameters $ Q $, $ P $ change, or if the  model parameter $d  $ is updated online (which is not uncommon), then this renders the learned policy $ \pi(x;\theta) $  useless. A new training data set $ \mathcal{D} $ now has to be generated using the modified MPC formulation and the approximate policy must be re-trained  and re-validated from scratch.  This is perhaps one of the biggest drawbacks of MPC policy approximation. In theory, this can of course be circumvented by training a parametric function  that is also a function of the weights and model parameters $ \pi_{}(p;\theta) $, where for example $ p := [x,\text{vec}(Q),\text{vec}(P),d]^{\mathsf{T}} $ and   $ \pi: \mathcal{P} \rightarrow \mathcal{U} $ with  $ \mathcal{P} $ being the the combined state and parameter space. However, this  makes the offline training problem even more expensive since the dimension of $ p $ can be extremely large. 
Despite the surge of interest in MPC policy approximation, its practical impact will depend on  addressing the pivotal challenge of  offline data generation. 

\subsection{Data Augmentation}
In machine learning literature, and more particularly, computer vision and image classification,  the issue of data-efficiency is addressed using a simple yet powerful  technique known as \enquote{Data Augmentation},  broadly defined as a \textit{strategy
	to artificially increase the number of training samples using computationally
	inexpensive transformations of the existing samples} \cite{tran2017bayesian,taylor2018improving,shorten2019survey}. When the data set is a collection of images (i.e. pixel-based data), data augmentation techniques include geometric transformations (such as rotate, translate, crop, flip, etc.), photometric transformations (such as colorize, saturation, contrast, brightness, etc.), and elastic transformation (such as deformation, shear, grid distortion, etc.) of existing images to artificially augment several additional training samples. 

Unfortunately, when it comes to data sets consisting of optimal state-action pairs sampled from an optimal policy such as $ \pi^*(x) $, the  existing data augmentation methods are not applicable. \ifpreprint For example, what do geometric/photometric/elastomeric transformations even mean for an optimal state-action pair $ (x_{i},u_{i}^*) $?! \fi 
By data augmentation for optimal policy observations, we mean the following: Given an optimal state-action pair $ (x_{i},u_{i}^*) $ observed by querying the expert policy $ \pi^*(x_{i}) $, we would like to augment additional data points using computationally inexpensive transformations that would tell us what the  optimal action $ {u}_{j}^*  $ would be for states $ x_{j}:= x_{i} + \Delta x $ in the vicinity of $ x_{i} $. If we can get  a reasonable approximation $ \hat{u}_{j} \approx  {u}_{j}^* $ without actually querying the expert, then the inexact optimal state-action pair $ (x_{j},\hat{u}_{j}) $ can be augmented to the set of demonstrations used for learning the policy. 
In other words, instead of simply learning from a given set of expert demonstrations, we would like to augment additional samples that are \emph{inferred} from the expert demonstrations, thereby enabling us to generalize to states not included in the original set of expert queries.

Noting that the MPC problem \eqref{Eq:MPC} is parametric in the initial condition, a  sensitivity-based data augmentation scheme amenable for optimal state-action pairs was recently proposed in  \cite{DK2021DataAug}, where it was shown that augmenting additional samples using parametric sensitivities only require the solution to a system of linear equations, which is computationally much cheaper than solving the optimization problem. 
In \cite{DK2021DataAug}, the error bound between the expert policy and approximate policy learned from the augmented data set was shown to depend  on the  size of the region $ \Delta x_{} $ used for data augmentation. 
Therefore, if we want a certain  desired accuracy, depending on the problem Lipschitz properties,  this limits the size of  $ \Delta x_{} $ that can be used for data augmentation. 
Building on our recent work \cite{DK2021DataAug}, the main contribution of this letter paper is an improved data augmentation scheme, where the  augmented data samples are enforced to be within an user-specified accuracy by using additional corrector steps. This would make the error bound independent of the size of $ \Delta x $, thereby enabling us to augment samples from a larger neighborhood $ \Delta x_{} $ without jeopardizing accuracy.

\section{Sensitivity-based data augmentation}
\subsection{Preliminaries}
For the sake of generality, we rewrite the MPC problem \eqref{Eq:MPC} in the standard parametric NLP form 
\begin{subequations}\label{Eq:NLP}
	\begin{align}
		\Pi(p): \quad	\min_{\mathbf{w}} &\; J(\mathbf{w},p)\\
		\textup{s.t.} &\; c(\mathbf{w},p) = 0, 
		\quad g(\mathbf{w},p)\le0
	\end{align}
\end{subequations}
where $ p \in \mathbb{R}^{n_{p}}$ includes current state of the 
system $ x_{i} $, vectorized tuning parameters $ Q,P $, system model parameters $ d $, or any other parameter that can change during online deployment. $ \mathbf{w} \in \mathbb{R}^{n_{w}}$ denotes the primal decision variables (with the  optimal action  $ u^*_{i} \in \mathbb{R}^{n_{u}} $ contained within  $\mathbf{w}^*$), $ J: \mathbb{R}^{n_{p}}\times \mathbb{R}^{n_{w}} \rightarrow \mathbb{R} $ denotes the objective function, $ c: \mathbb{R}^{n_{p}}\times \mathbb{R}^{n_{w}} \rightarrow \mathbb{R}^{n_{c}} $ and $ g: \mathbb{R}^{n_{p}}\times \mathbb{R}^{n_{w}} \rightarrow \mathbb{R}^{n_{g}} $ denotes the set of equality and inequality constraints, respectively. 

The Lagrangian of the optimization problem \eqref{Eq:NLP} is given by
$	\mathcal{L}(\mathbf{w},\lambda,\mu,p) = J(\mathbf{w},p) + \lambda^{\mathsf{T}} c(\mathbf{w},p) + \mu^{\mathsf{T}}g(\mathbf{w},p)
$ 
where $ \lambda \in \mathbb{R}^{n_{c}} $ and $ \mu \in \mathbb{R}^{n_{g}}$ are the dual variables. 
For the inequality constraints, we define $ g_{\mathbb{A}}   \subseteq g$ as the set of active inequality constraints (i.e. $ g_{\mathbb{A}}(\mathbf{w},p) =0 $) and assume strict complimentarity (i.e., $ \mu _{\mathbb{A}}>0 $). The first order necessary conditions of optimality  can be denoted compactly as,
\begin{align}\label{Eq:KKTcompact}
	\varphi(\mathbf{s}(p),p) &:= \begin{bmatrix}
		\nabla\mathcal{L}(\mathbf{w},\lambda,\mu,p) \\
		c(\mathbf{w},p) \\
		g_{\mathbb{A}}(\mathbf{w},p) 
	\end{bmatrix} = 0 \\
\mu_{\mathbb{A}}&>0, \mu_{\bar{\mathbb{A}}}=0, g_{\bar{\mathbb{A}}}\nonumber (\mathbf{w},p)<0
\end{align} 
Any primal-dual vector $ \mathbf{s}^*(p) := [\mathbf{w}^*,\lambda^*,\mu^*]^{\mathsf{T}} $ is called a KKT-point if it satisfies the first order necessary conditions of optimality \eqref{Eq:KKTcompact}. 
\begin{assume}\label{asm:TwiceDifferentiable}
	The cost and constraints $ J(\cdot,\cdot)$, $ c(\cdot,\cdot) $, and $ g(\cdot,\cdot) $ of the NLP problem $\Pi(p)$ are twice continuously differentiable in the neighborhood of the KKT point $ \mathbf{s}^*(p) $. 
\end{assume}
\begin{assume}\label{asm:UniqueMinima}
	Linear independence constraint qualification, second order sufficient conditions and strict complementarity holds for $\Pi(p)$.  
\end{assume}
The above assumption implies that  the primal-dual  solution vector $ \mathbf{s}^*(p_{i}) $ is a unique local minimizer of $ \Pi(p) $. 	Given Assumptions~\ref{asm:TwiceDifferentiable} and \ref{asm:UniqueMinima}, it was shown in \cite{fiacco1976sensitivity} that for parametric perturbations $\Delta p$ in the neighborhood of $ p_{i} $,
there exists a unique, continuous, and differentiable vector
function $ \mathbf{s}(p_{i} +\Delta p ) $ which is a KKT point satisfying LICQ
and SOSC for $ \Pi(p_{i}+ \Delta p ) $, and that the solution vector satisfies $ 	\| \mathbf{s}^*(p_{i} + \Delta p )  - \mathbf{s}^*(p_{i})\| \le  L_{s} \|\Delta p \| $, where the notation $ \|\cdot\| $ by default denotes  Euclidean norm.

\begin{theorem}[\cite{DK2021DataAug}]\label{thm:DataAugment}
	Given Assumptions~\ref{asm:TwiceDifferentiable} and \ref{asm:UniqueMinima}, 
let $ u_{i}^* $ be the optimal policy obtained by querying  $ \Pi(p_{i}) $ for some $ p_{i} \in \mathcal{P} $. Then additional  samples $\{(p_{i} + \Delta p _{j}, u_{i}^*+\Delta u_j)\}_{j} $ in the neighborhood of $ p_{i} $ with the same active constraint set $ g_{\mathbb{A}} $ can be augmented to the data set without querying $ \Pi(p_{i} + \Delta p _{j}) $, where 
$\Delta u_{j} \subset \Delta \mathbf{s}_{j}$  is given by the linear predictor $  \Delta\mathbf{s}_{j} = \mathcal{H}_{i} \Delta p _{j} $.

\end{theorem}
\begin{proof}
Taylor series expansion of $ \mathbf{s}^*(p) $ around $ p_i $, gives
\begin{align}\label{Eq:Taylor}
	\mathbf{s}^*(p_i+\Delta p_{j} ) =\mathbf{s}^*(p_{i}) + \frac{\partial \mathbf{s}^*}{\partial p} \Delta p_{j}  + \mathcal{O}(\|\Delta p_{j} \|^2)
\end{align}
Since  $ \mathbf{s}^*(p) $ satisfies \eqref{Eq:KKTcompact} for any $ p $ in the neighborhood of $ p_{i} $, using \textcolor[rgb]{0,0,0}{the }implicit function theorem, we have
\begin{align*}
	&\frac{\partial}{\partial p} \varphi(\mathbf{s}^*(p),p) \bigg|_{p=p_{i}}  = \frac{\partial  \varphi}{\partial \mathbf{s}}  \frac{\partial\mathbf{s}^*}{\partial p} + \frac{\partial \varphi}{\partial p} = 0\\
	\Rightarrow &   \frac{\partial\mathbf{s}^*}{\partial p}  = -\left[\frac{\partial  \varphi}{\partial \mathbf{s}}  \right]^{-1}\frac{\partial \varphi}{\partial p} =: \mathcal{H}_{i}
\end{align*}
\ifpreprint
where,
\[ \frac{\partial  \varphi}{\partial \mathbf{s}} :=  	\begin{bmatrix}
	\nabla ^2_{\mathbf{w}\mathbf{w}} \mathcal{L}&\nabla_{\mathbf{w}}{c} & \nabla_{\mathbf{w}}{g}_{\mathbb{A}}\\ 
	\nabla_{\mathbf{w}}{c}^{\mathsf{T}} & 0  & 0 \\ 
	\nabla_{\mathbf{w}}{g}_{\mathbb{A}}^{\mathsf{T}}& 0 &0
\end{bmatrix}, \frac{\partial \varphi}{\partial p}:=\begin{bmatrix}
\nabla ^2_{\mathbf{w}p} \mathcal{L}\\ 
\nabla_{p}{c}^{\mathsf{T}}\\ 
\nabla_{p}{g}_{\mathbb{A}}^{\mathsf{T}}
\end{bmatrix} \]
For the sake of notational simplicity we will drop the subscripts such that $ \nabla(\cdot) $ indicates $ \nabla_{\mathbf{w}}(\cdot) $ by default, unless otherwise specified explicitly. 
\fi 
Substituting this in \eqref{Eq:Taylor} and ignoring the higher order terms,  we get the linear predictor
\begin{equation}\label{Eq:SensitivityUpdate1}
 \textcolor[rgb]{0,0,0}{\mathbf{s}^*(p_i+\Delta p_{j} )  \approx	\hat{\mathbf{s}}(p_{i}+\Delta p_{j} ) }  = \mathbf{s}^*(p_{i})  + \underbrace{\mathcal{H}_{i}\Delta p_{j} }_{:= \Delta \mathbf{s}_{j}}
\end{equation}
which contains the inexact optimal action $ \hat{u}_{j}  = u_{i}^* + \Delta u_{j}$.
\end{proof}
\begin{rem}[Matrix factorization and inverse]The main computation in the linear predictor involves getting $ \mathcal{H}_{i} $. Notice that the first term $ \frac{\partial  \varphi}{\partial \mathbf{s}}  $ is nothing but the KKT matrix, which is already  factorized when solving the NLP $ \Pi(p_{i}) $ using e.g. \texttt{IPOPT} \cite{pirnay2012optimal}. More importantly,
augmenting any number of data samples using a single data sample $ (p_{i}, u_{i}^*)$ requires computing  $ \mathcal{H}_{i} $  only  once! That is, in the simplest case, using only a single NLP solve at $ x_{i} $, and a single linear solve (to compute $ \mathcal{H}_{i} $), we can efficiently augment several data samples  $ \Delta s_{j} = \mathcal{H}_{i}\Delta p_j $. 
\end{rem}

\ifpreprint
\begin{figure}
	\centering
	\includegraphics[width=0.6\linewidth]{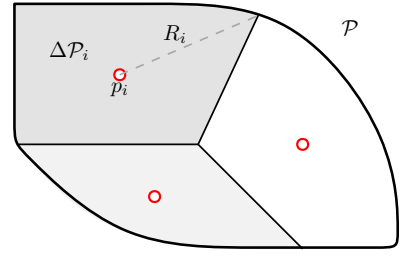}
	\caption{Schematic representation of  the inexact solution manifold with $ N_{s} = 3 $ points where the NLP is solved exactly (denoted by red circle), and parameter space $ \mathcal{P} $ divided into the regions $ \Delta \mathcal{P}_{i} $ around each $ p_{i} $, which are used for data augmentation.} \label{Fig:PWAmanifold}
\end{figure}
\fi
\begin{figure}
	\centering
	\includegraphics[width=0.9\linewidth]{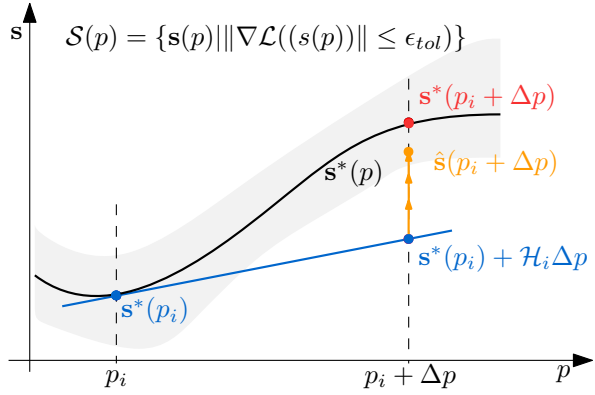}
	\caption{Schematic illustration of the proposed sensitivity-based data augmentation. Blue line indicate the linear predictor step, and the orange arrows indicate the corrector steps. The set $ \mathcal{S}(p) $ is shown in gray. }\label{Fig:DataAug}
\end{figure}
To understand the effect of augmenting the data set with inexact samples, consider the case where the true solution manifold of  $ \Pi(p) $ is given by $ \mathbf{s}^*(p) $ for all $ p \in \mathcal{P} $, which  is sampled at $ N_{s} $ discrete points $ \{p_{i}\}_{i=1}^{N_{s}} \in \mathcal{P} $.  Assume that there exists a piece-wise affine inexact solution manifold, denoted by  $ \hat{\mathbf{s}}(p) $ for all $ p \in \mathcal{P} $ given by the linear predictor \eqref{Eq:SensitivityUpdate1}  within a neighborhood $ \Delta \mathcal{P}_{i} $ around each $ p_{i} $
 for all $ i = 1, \dots,N_{s} $. We make the following assumption regarding the neighborhood $ \Delta \mathcal{P}_{i} $.
\begin{assume}\label{asm:inexactManifold}
For all $ i = 1, \dots,N_{s} $, the neighborhood  $ \Delta \mathcal{P}_{i} $ is chosen
around each $ p_{i} $  such that $ p_{i} \in \langle\Delta \mathcal{P}_{i}\rangle $, $ \bigcup_{i=1}^{N_{s}} \Delta \mathcal{P}_{i} = \mathcal{P}$, and that the active constraint set remains the same within each neighborhood  $ \Delta \mathcal{P}_{i} $ \textcolor[rgb]{0,0,0}{, and the solution vector satisfies $ 	\| \mathbf{s}^*(p_{i} + \Delta p )  - \mathbf{s}^*(p_{i})\| \le  L_{s} \|\Delta p \| $.}
\end{assume}
We construct the idea of such a  piece-wise affine inexact manifold $ \hat{\mathbf{s}}(p) $ given by the linear predictor \eqref{Eq:SensitivityUpdate1}, so that  the  augmented data points are considered to be sampled from this inexact manifold.
 \textcolor[rgb]{0,0,0}{Learning the policy then involves  fitting a parametric function $ \pi (x,\theta) $ to the  augmented training data set $ \hat{u}(p) \subset \hat{\mathbf{s}}(p)  $,
such that
\begin{equation}\label{Eq:LSE}
	\theta_1 = \arg \min_{\theta} \sum_{i=1}^{N_{s}+M} \|\pi(p_{i};\theta) - \hat{u}(p_{i})\|^2
 \end{equation}}
\begin{assume}\label{asm:RichParametrization}
	The functional form of $ \pi(p;\theta) $ has sufficiently rich parametrization and $ \exists \; \hat{\theta} $ such that \textcolor[rgb]{0,0,0}{$ \hat{u}(p) = \pi(p;\hat{\theta}) $.}
\end{assume}

\begin{theorem}[\cite{DK2021DataAug}]\label{thm:MaxD}
		Consider a problem with the same setup as in Theorem~\ref{thm:DataAugment}, where the base data set $ \mathcal{D}^0 $ with $ N_{s} $ samples is obtained by querying $ \Pi(p) $, which is  augmented with  $ M $ inexact samples from $ \hat{u}(p) $ using \eqref{Eq:SensitivityUpdate1}.  Under Assumptions~\ref{asm:inexactManifold} and \ref{asm:RichParametrization}, 
		\begin{equation}\label{Eq:Errorbound0}
			\|\pi(p;\theta_1)-\pi^*(p)\| \le \max\left( \left\{L_{p_{i}}R_{i}^2\right\}_{i=1}^{N_{s}}\right) 
		\end{equation} in probability as $ M \rightarrow \infty $  holds if $ \theta_1 $ is a consistent estimator of \eqref{Eq:LSE}, with $ R_{i} :=  \sup_{\Delta p \in \Delta\mathcal{P}_{i}}\|\Delta p\|$.
	\end{theorem}
	\begin{proof}
		Due to the continuity and differentiability of $ \mathbf{s}^*(p_{i}) $\cite{fiacco1976sensitivity}, the following holds for all $  \Delta p \in \Delta \mathcal{P}_{i}  $, 
		\begin{align*}
			\|\hat {\mathbf{s}}(p_{i} + \Delta p) - \mathbf{s}^*(p_{i} + \Delta p)\|  \leq  L_{p_{i}}\|\Delta p\|^2 \leq  L_{p_{i}}R_{i}^2
		\end{align*}
		Aggregating over all the regions  results in the inequality 	\begin{equation*}
			\|\hat {\mathbf{s}}(p) - \mathbf{s}^*(p)\|  \leq \max\left( \left\{L_{p_{i}}R_{i}^2\right\}_{i=1}^{N_{s}}\right) \quad \forall p \in \mathcal{P}
		\end{equation*}
 \textcolor[rgb]{0,0,0}{	Since ${u} \subset  {\mathbf{s}}(p) $, we have $
		\|\hat {u}(p) - u^*(p)\|  \leq 	\|\hat {\mathbf{s}}(p) - \mathbf{s}^*(p)\| $,}
which	quantifies the maximum deviation between the optimal policy and the augmented inexact samples. If $ \theta_1 $ is a consistent estimator of \eqref{Eq:LSE}, then inequality \eqref{Eq:Errorbound0} follows under Assumption~\ref{asm:RichParametrization}.
	\end{proof}


From this we can see that the error due to learning a policy based on augmented data samples \textcolor[rgb]{0,0,0}{depends} on the problem Lipschitz properties $ L_{p_{i}} $, and the size of the neighborhood $ \Delta p_{i} $ used to augment the data samples\ifpreprint (cf. Fig.~\ref{Fig:PWAmanifold})\fi. If one were to use a large neighborhood $ \Delta p_{i} $ in order to reduce the number of offline NLP computations, then this can potentially affect the performance of the learned policy.

\begin{algorithm}[t]
	\vspace{1mm}
	\caption{Improved predictor-corrector sensitivity-based data augmentation.}
	\label{alg:SensitivityTraining2}
	\begin{algorithmic}[1]
		\Require $ \Pi(p) $, $ \mathcal{P} $, $ \mathcal{D} = \emptyset$, $ \epsilon_{tol}  $
		\algrulehor
		\For {$ i = 1,\dots,N_{s} $}
		\State Sample $ p_{i} \in \mathcal{P} $
		\State $ \mathbf{s}^*(p_{i}) \leftarrow $ Solve $ \Pi(p_{i}) $
		\State Extract $ u_{i}^* $ from the solution vector $  \mathbf{s}^*(p_{i}) $
		\State $\mathcal{D} \leftarrow \mathcal{D} \cup \{(p_{i},u_{i}^*)\}$
		\State $ \mathcal{H}_{i} \leftarrow  -\left[\frac{\partial  \varphi(\mathbf{s}^*(p_{i}) ,p_{i})}{\partial \mathbf{s}}  \right]^{-1}\frac{\partial \varphi(\mathbf{s}^*(p_{i}) ,p_{i})}{\partial p} $
		\For {$ j = 1,\dots,m $}
		\State Sample $ \Delta p_{j} \in \Delta\mathcal{P}_{i} $ in the neighborhood of $ p_{i}$
		\State  $  \hat{\mathbf{s}}(p_{i}+\Delta p_{j})  \leftarrow\mathbf{s}^*(p_{i})  + \mathcal{H}_{i}\Delta{p}_{j}$ \Comment predictor
		\While{$ \|
			\nabla \mathcal{L}(\hat{\mathbf{s}}(p_{i}+\Delta p_{j}))\|>\epsilon_{tol} $} \Comment corrector
		\State $  \hat{\mathbf{s}}(p_{i}+\Delta p_{j})  \leftarrow\begin{bmatrix}
			\hat{\mathbf{w}}_{ij}\\
			0\\
			0
		\end{bmatrix}  -  \mathcal{M}_{ij}^{-1} \mathcal{Q}_{ij}$
		\EndWhile
		\State Extract $ \hat{u}_{j}^* $ from the solution vector $  \mathbf{\hat s}^*(p_{i}+ \Delta{p}_{j}) $
		\EndFor
		\EndFor
		
		\algrulehor
		\Ensure $ \mathcal{D} $
	\end{algorithmic}
\end{algorithm}

\subsection{Data augmentation with user-enforced accuracy}
In order to address this issue, we now propose an improved data augmentation scheme, where in addition to the linear predictor,  corrector steps are taken to reduce the approximation error.
\ifpreprint
Notice that the KKT conditions can equivalently be expressed in the following matrix notation
\begin{equation*}
		\begin{bmatrix}
		\nabla ^2 \mathcal{L}&\nabla{c} & \nabla{g}_{\mathbb{A}}\\ 
		\nabla{c}^{\mathsf{T}} & 0  & 0 \\ 
		\nabla{g}_{\mathbb{A}}^{\mathsf{T}}& 0 &0
	\end{bmatrix}\begin{bmatrix}
0\\
	 \lambda^*\\
 \mu_{\mathbb{A}}^* 
\end{bmatrix} + \begin{bmatrix}
\nabla J	\\
c\\
g_{\mathbb{A}} 
\end{bmatrix}
\end{equation*}
Adding this with the predictor step results in 
\begin{equation*}
	\begin{bmatrix}
		\nabla ^2 \mathcal{L}&\nabla{c} & \nabla{g}_{\mathbb{A}}\\ 
		\nabla{c}^{\mathsf{T}} & 0  & 0 \\ 
		\nabla{g}_{\mathbb{A}}^{\mathsf{T}}& 0 &0
	\end{bmatrix}\begin{bmatrix}
		\Delta w \\
	\Delta \lambda+	\lambda^*\\
\Delta \mu_{\mathbb{A}}+	\mu_{\mathbb{A}}^* 
	\end{bmatrix} + \begin{bmatrix}
		\nabla J	\\
		c\\
		g_{\mathbb{A}}
	\end{bmatrix} +\begin{bmatrix}
	\nabla ^2_{\mathbf{w}p} \mathcal{L}\\ 
	\nabla_{p}{c}^{\mathsf{T}}\\ 
	\nabla_{p}{g}_{\mathbb{A}}^{\mathsf{T}}
\end{bmatrix} \Delta p = 0
\end{equation*}
Setting $ \Delta p = 0 $ gives us the SQP correction step
\fi
\begin{equation}\label{Eq:NewtonCorrector}
	\begin{bmatrix}
		H&\nabla{c} & \nabla{g}_{\mathbb{A}}\\ 
		\nabla{c}^{\mathsf{T}} & 0  & 0 \\ 
		\nabla{g}_{\mathbb{A}}^{\mathsf{T}}& 0 &0
	\end{bmatrix}\begin{bmatrix}
		\Delta w\\
		\Delta \lambda+ \lambda^*\\
		\Delta \mu_{\mathbb{A}} +\mu_{\mathbb{A}}^*
	\end{bmatrix} + \begin{bmatrix}
		\nabla J	\\
		c\\
		g_{\mathbb{A}}
	\end{bmatrix} = 0
\end{equation}
 \textcolor[rgb]{0,0,0}{Expanding $ \hat{\mathbf{s}}(p_{i} + \Delta p _{j}) := [\hat{\mathbf{w}}_{ij},\hat{\lambda}_{ij},\hat{\mu}_{ij}]^{\mathsf{T}} $, the approximate solution is updated as }
\begin{equation}\label{key}
	\begin{bmatrix}
		\hat{\mathbf{w}}_{ij}\\
		\hat{\lambda}_{ij}\\
		\hat{\mu}_{ij}
	\end{bmatrix} = 	\begin{bmatrix}
	\hat{\mathbf{w}}_{ij}\\
0\\
0
\end{bmatrix} - {\underbrace{\begin{bmatrix}
H & \nabla c & \nabla g_{\mathbb{A}}\\
\nabla c^{\mathsf{T}} & 0 & 0\\
\nabla g_{\mathbb{A}}^{\mathsf{T}} &0&0
\end{bmatrix}}_{:= \mathcal{M}_{ij}}}^{-1} \underbrace{\begin{bmatrix}
\nabla J \\
c\\
g_{\mathbb{A}}
\end{bmatrix}}_{:= \mathcal{Q}_{ij}}
\end{equation}
where $ H $ (Hessian of the Lagrangian), $ \nabla c $, $ \nabla g_{\mathbb{A}} $,$ \nabla J $, $ c $, and $ g_{\mathbb{A}} $ that make up $ \mathcal{M}_{ij} $ and $ \mathcal{Q}_{ij} $ are all evaluated at $\hat{\mathbf{s}}(p_{i} + \Delta p _{j})  $. 
The corrector steps are taken, until the optimality residual defined by 
$\| \nabla \mathcal{L}(\hat{\mathbf{s}}(p_{i} + \Delta p _{j}) )\|$
is less than some user defined tolerance $ \epsilon_{tol} $.
 \textcolor[rgb]{0,0,0}{ If the active constraint set $ g_{\mathbb{A}} $ remains the same (cf. Assumption~\ref{asm:inexactManifold}), then the corrector step corresponds  to a Newton direction on the KKT condition (cf. \eqref{Eq:NewtonCorrector}). See   \cite{kungurtsev2017predictor} and the references therein for  detailed description of the corrector step.} The proposed data augmentation scheme with predictor-corrector steps is summarized in Algorithm~\ref{alg:SensitivityTraining2}.

\begin{theorem}\label{thm:MaxD2}
	Consider a problem with the same setup as in Theorem~\ref{thm:DataAugment}, where the base data set $ \mathcal{D}^0 $ with $ N_{s} $ samples is obtained by querying $ \Pi(p) $, which is  augmented with  $ M = N_{s}m $ inexact samples from $ \hat{\mathbf{s}}(p) $ using Algorithm~\ref{alg:SensitivityTraining2}.  Under Assumption~\ref{asm:RichParametrization}, 
		\begin{equation}\label{Eq:Errorbound}
		\|\pi(p;\theta_1)-\pi^*(p)\| \le r(p)
	\end{equation} holds in probability as $ M \rightarrow \infty $  if $ \theta_1 $ is a consistent estimator of \eqref{Eq:LSE}, where $ r(p) :=  \sup_{\mathbf{s}(p) \in \mathcal{S}(p)}\|\mathbf{s}(p) - \mathbf{s}^*(p)\|$, with $ \mathcal{S}(p):= \{ \mathbf{s}(p) | \| \nabla \mathcal{L}(\hat{\mathbf{s}}(p) ) \| \le \epsilon_{tol}   \} $.
\end{theorem}
\begin{proof}
	At the solution manifold $ \mathbf{s}^*(p) $, we have $ \nabla\mathcal{L}(\mathbf{s}^*(p)) =0$. From the continuity of the cost and constraints,  $ \exists \; \mathcal{S}(p):= \{ \mathbf{s}(p) \;| \;\| \nabla \mathcal{L}({\mathbf{s}}(p) ) \| \le \epsilon_{tol}   \}  $   containing the solution manifold $ \mathbf{s}^*(p)  $ in its interior (cf. Fig.~\ref{Fig:DataAug})  for all $ p $. Solving the corrector steps until the bound on the optimality residual satisfies $ \| \nabla \mathcal{L}(\hat{\mathbf{s}}(p) ) \| \le \epsilon_{tol}    $ (cf. line 10 in Algorithm~\ref{alg:SensitivityTraining2}) implies that the $ \hat{\mathbf{s}}(p)  \in \mathcal{S}(p)$. Defining the maximum distance from the exact solution manifold $ \mathbf{s}^*(p) $ to the boundary of the set $ \mathcal{S}(p) $ for any $ p $ as \[  r(p) :=  \sup_{\mathbf{s}(p) \in \mathcal{S}(p)}\|\mathbf{s}(p) - \mathbf{s}^*(p)\| \] results in the inequality $ 	\|\hat{\mathbf{s}}(p) - \mathbf{s}^*(p)\| \le r(p) $. 	Since ${u} \subset  {\mathbf{s}}(p) $, we have $
	\|\hat {u}(p) - u^*(p)\|  \leq 	\|\hat {\mathbf{s}}(p) - \mathbf{s}^*(p)\| $.
	Furthermore, from Assumption~\ref{asm:RichParametrization}, we have that
	\begin{align*}
		\|\hat{u}(p) - u^*(p)\| =& \| \pi(p;\hat{\theta}) - \pi^*(p)\|\leq r(p)
	\end{align*}
Our result follows,	if $ \theta_1 $ is a consistent estimator of \eqref{Eq:LSE}.  
\end{proof}

The bound $ r(p) $ clearly depends on $ \epsilon_{tol} $, which can be controlled by the user, as opposed to the bound in Theorem~\ref{thm:MaxD}, which depends on the distance  from the original sample.

\begin{rem}[Augmented samples]
The augmented samples can be approximated from the NLP sample, or another previously augmented sample (similar to a path-following scheme). In either case,  corrector steps are taken until the optimality residual is less than $ \epsilon_{tol} $. Therefore, Theorem~\ref{thm:MaxD2} is valid in both cases. 
\end{rem}
\begin{rem}[Active constraint set]
 \textcolor[rgb]{0,0,0}{	Note that if a parturbation in $ p $ changes the active constraint set $ g_{\mathbb{A}} $, this would  induce  non-smooth points in the solution manifold $ \mathbf{s}^*(p) $. } Capturing the non-smooth points in $ \mathbf{s}^*(p) $ using piecewise-linear prediction manifolds requires solving a predictor-corrector QP  \cite{bonnans1998optimization,kungurtsev2017predictor}. However, in the context of data augmentation, a simpler alternative  could be to discard any samples $ p_{i}+\Delta p  $ that induce a change in the active constraint set.
 \textcolor[rgb]{0,0,0}{In other words, Assumption~\ref{asm:inexactManifold}  limits the size of the neighborhood around which additional samples can be augmented.}
\end{rem}


  \begin{table*}[h]
	\caption{Example 1 - Comparison of the total CPU time required  to generate the different data sets, and the  maximum error. 
	}\label{tb:CPU time}
	\centering
	\small
	\begin{tabular}{ |c||c||c|c||c|c| } 
		\hline
		\multirow{2}{*}{}& full NLP&  \multicolumn{2}{c||}{ predictor only} &  \multicolumn{2}{c|}{predictor-corrector}\\
		\cline{2-6}
		& CPU time [s] & CPU time [s] & Max Error & CPU time [s] & Max Error  \\
		\hline
		Case 1 & 466.42 &4.92& 1.34&11.95&0.0074\\ 
		Case 2  & 456.21&0.45 &12.51& 8.67 & 0.028\\ 
		Case 3  &  542.73& 0.15&79.88&12.96& 0.018\\ 
		\hline
	\end{tabular}
\end{table*}
\ifpreprint
\begin{figure*}[t]
		\begin{subfigure}{0.5\textwidth}
		\centering
		\includegraphics[width=0.95\linewidth]{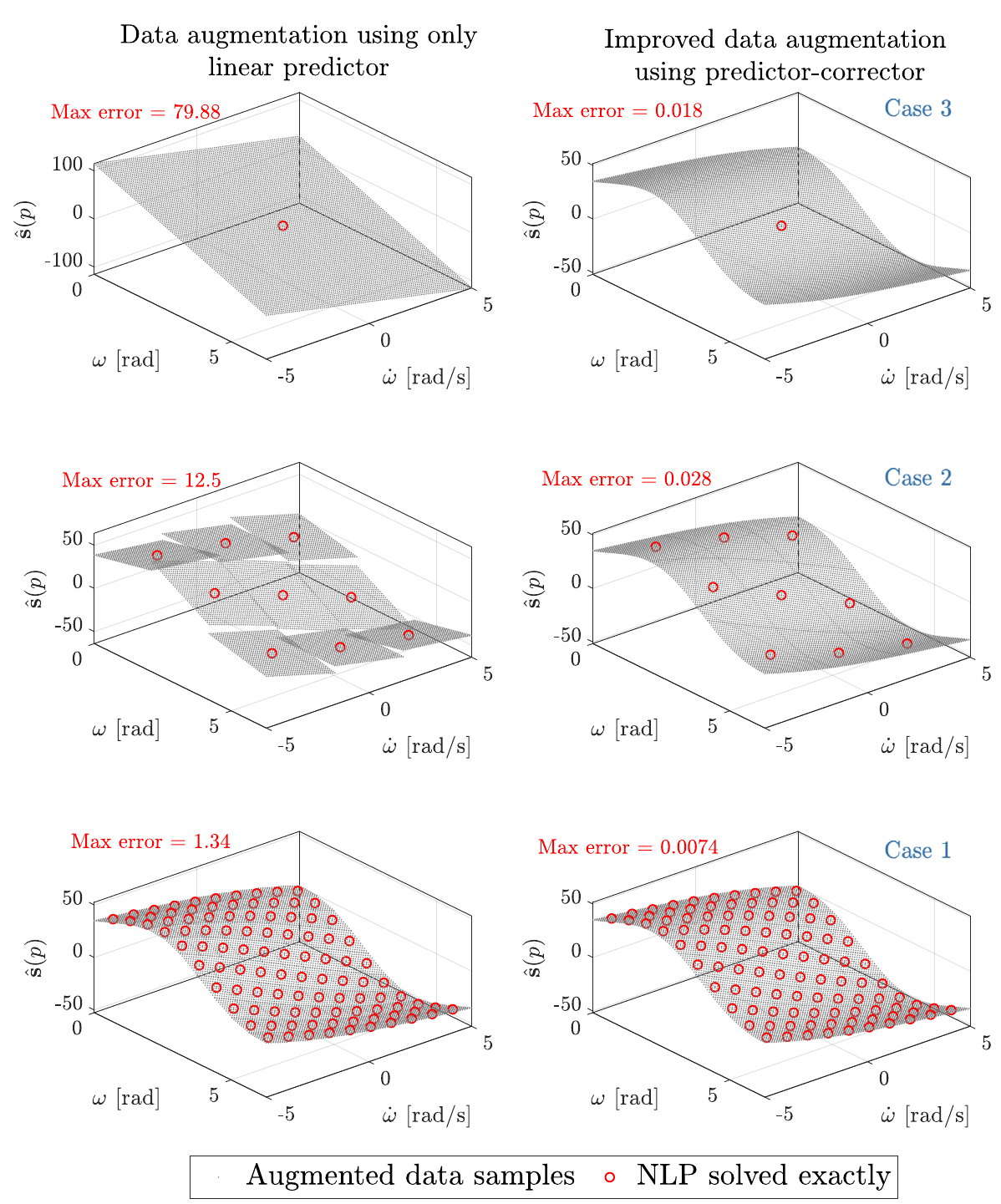}
		\caption{}\label{Fig:pendulumManifold}
	\end{subfigure}
	\begin{subfigure}{0.5\textwidth}
		\centering
	\includegraphics[width=0.95\linewidth]{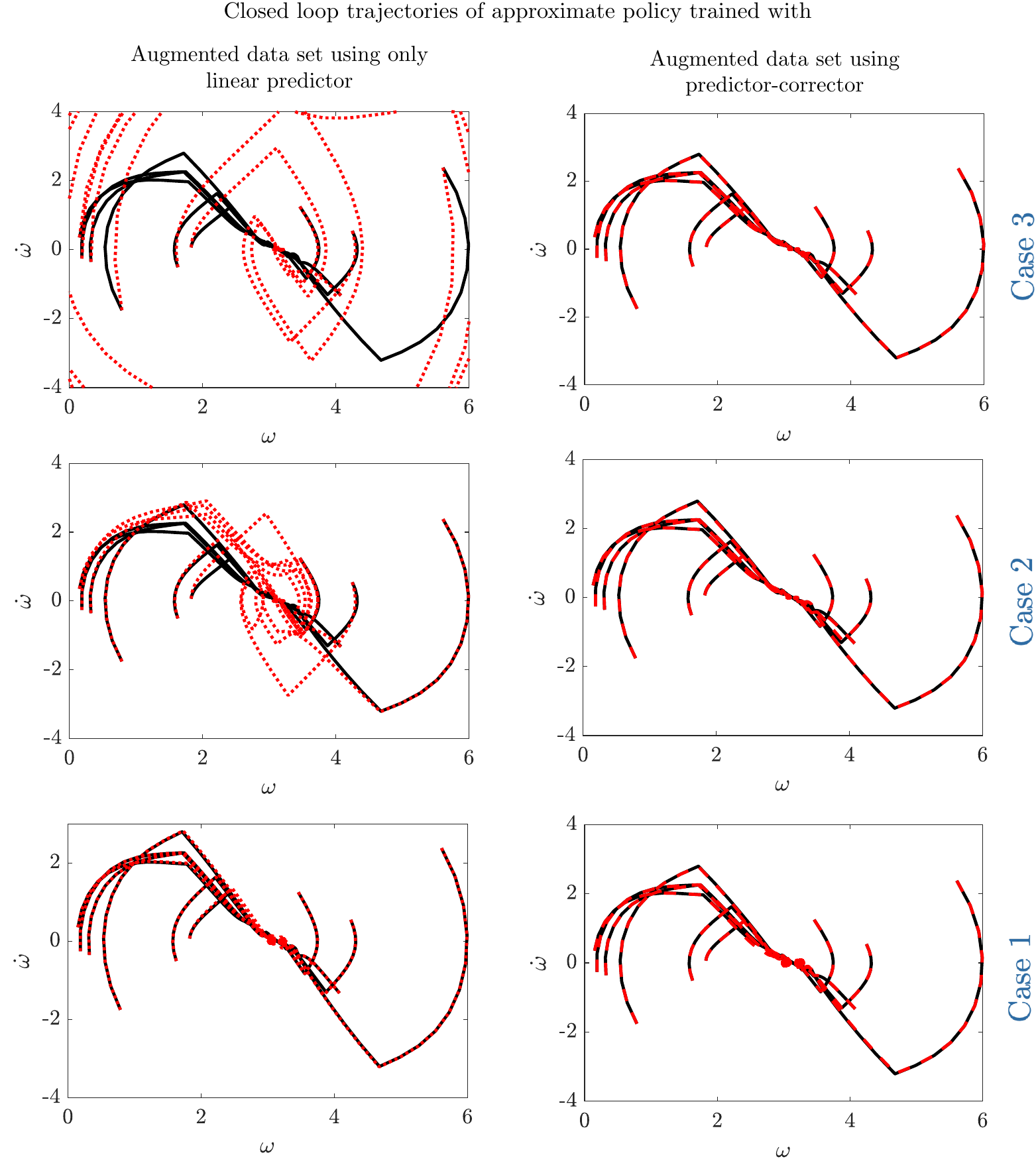}
	\caption{}\label{Fig:pendulumSim}
	\end{subfigure}

\caption{Example 1 - (a) Data set with optimal state-action data pairs for the inverted pendulum example.
	Red circles indicate samples where the MPC problem is solved exactly, and gray dots indicate samples that are augmented using only the linear predictor \cite{DK2021DataAug} (left subplots), and the proposed approach using additional corrector steps (right subplots). (b) Closed-loop trajectories of the approximate policy trained with augmented data set using only the linear predictor \cite{DK2021DataAug} (red dotted lines in left subplots), and with augmented data set using proposed approach (red dashed lines in right subplots), corresponding to the data sets shown in Fig.~\ref{Fig:pendulumManifold}, which are benchmarked against the true MPC policy (solid black lines) for various starting points.}
\end{figure*}
\fi

\section{Numerical experiments}

\ifpreprint
\subsection{Example 1: Inverted Pendulum}
\fi

 \textcolor[rgb]{0,0,0}{We demonstrate the performance of the data augmentation framework using the benchmark inverted pendulum problem with the nonlinear dynamics $ml^2 \ddot{\omega}  = u - b\dot{\omega} - mgl\sin\omega$ described by the two states, angle $ \omega  $, and angular velocity $ \dot{\omega} $.} We consider the case where the expert policy is given by a \textcolor[rgb]{0,0,0}{nonlinear model predictive controller with an objective to drive the pendulum to its inverted position.}  The input $ u $ is the torque, \textcolor[rgb]{0,0,0}{which is used to maintain the pendulum at its inverted position of $ \omega = 3.14 $ rad, $ \dot{\omega} = 0 $ .} In this example, we have $ p = [\omega,\dot{\omega}]^{\mathsf{T}} $ and $ \mathcal{P} = [0,2\pi]\times [-5,5]$. Note that we deliberately  choose this simple example to facilitate visualization of the exact and inexact solution manifolds to illustrate the data augmentation schemes.

To approximate the MPC policy, we wish to generate  $ 100\times100 $ state-action data pairs using a grid-based sampling strategy to evenly cover the parameter space $ \mathcal{P} $, which would normally take $10^4$ offline NLP computations. In each of the following cases, we also solve the full NLP to compute $ \pi^*(p) $ at each sample to serve as a benchmark.
The NLP problems $ \Pi(p) $ are solved using\texttt{ IPOPT} \ifpreprint \cite{wachter2006ipopt}\fi with \texttt{MUMPS}
linear solver. The optimization problem and the NLP sensitivities were formulated using \texttt{CasADi v3.5.1} \cite{andersson2019CasADi}. All augmented samples are based on the sensitivity updates from the exact sample obtained by solving the NLP. 
All computations were performed on a 2.6 GHz
processor with 16GB memory. 

\paragraph*{Case 1: Solve 100 offline NLP problems}
 In this case, we queried the NMPC expert to generate $ N_{s} =  100 $ samples using a  $ 10\times10 $ grid as shown in Fig.~\ref{Fig:pendulumManifold} in red circles. 
Using each data sample, we further augment additional 100 data samples around each $ p_{i} $, leading to a total of $ M= $ 10$ ^4 $ augmented data samples using only the linear predictor as done in \cite{DK2021DataAug}, as well as the proposed improved predictor-corrector data augmentation scheme.  
These are shown in Fig.~\ref{Fig:pendulumManifold} in gray dots, where it can be seen that both the data augmentation schemes are able to sufficiently capture the solution manifold. 
The maximum error due to approximation when using only the linear predictor is 1.34, whereas the maximum error was 0.0074 with the improved data augmentation. 

\paragraph*{Case 2: Solve 9 offline NLP problems}
We then consider the case where we queried the NMPC expert only $ N_{s} =  9 $ times  using a  $ 3\times3 $ grid as shown in Fig.~\ref{Fig:pendulumManifold} in red circles. 
Using each data sample, we further augment additional 1089 data samples around  each sample, leading to a total of 9081 augmented data samples using only $ 9 $ offline NLP computations. Since the size of  the neighborhood increases, the inexact solution manifold using only the linear predictor is now visbily different from the exact solution manifold (with 9 piecewise affine manifolds), with the maximum error of 12.5. On the other hand, when using the improved data augmentation scheme, the maximum error is only 0.028. This shows that despite the large neighborhood used for data augmentation, we can get the generate data samples with a desired accuracy at only marginally increased computational cost (cf. Table~\ref{tb:CPU time}).

\paragraph*{Case 3: Solve only 1 offline NLP problem}
We then consider the extreme case, where we only solve $ N_{s} =1 $ offline NLP problem as shown in Fig.~\ref{Fig:pendulumManifold}. Using this single NLP computation, we augment 10000 data samples using a $ 100\times 100 $ grid. Using only the linear predictor, naturally we get a single affine manifold, leading to large approximation errors. However, using the improved data augmentation scheme, we are still able to augment data samples with a desired accuracy, where the maximum error is only 0.018, as opposed to 79.88. 

The total CPU time to generate the data samples using the two data augmentation schemes, as well as solving all the data points exactly along with the maximum error are  summarized in Table~\ref{tb:CPU time}, which also shows the  trade-off between  computation cost and accuracy. \ifpreprint 
We then learn the approximate policy using the data sets from the six different cases shown in Fig.~\ref{Fig:pendulumManifold}.  For the approximate policies in each case, we choose a generalized regression neural network (a variant of the radial basis function network) that was trained using the  \texttt{newgrnn} function in \texttt{ MATLAB v2020b}. The hyperparameters of $ \pi(p,\theta) $ were the same across all the cases, and the only variation was the data set used for training. The closed-loop trajectories of the approximate policies compared to the true MPC policy starting from different  initial conditions are shown in Fig.~\ref{Fig:pendulumSim}. Here, red dotted lines denote the closed-loop trajectories of the the approximate policy learned using augmented data set with only the linear predictor, red dashed lines denote the closed-loop trajectories of the approximate policy learned using augmented data set using the proposed predictor-corrector based data augmentation, summarized in Algorithm~\ref{alg:SensitivityTraining2}, and solid black lines denote the true MPC policy that serves as a benchmark. These results clearly show that the control policy learned using only the augmented data set with linear predictor steps results in reasonable performance if the neighborhood used for data augmentation is small. As the size of the neighborhood used for data augmentation increases, the approximation error increases (cf. Theorem~\ref{thm:MaxD}), resulting in poorer closed-loop performance. However, using the proposed data augmentation scheme with user enforced accuracy, we are able to learn a satisfactory control policy, even when a large neighborhood is used for data augmentation. This clearly demonstrates the benefit of using the improved data augmentation scheme, which enables us to use a larger neighborhood for data augmentation at the cost of a few additional linear solves (corrector steps). 
\fi

\ifpreprint
\subsection{Example 2: Quadcopter control}
The proposed data augmentation method is also applied to a quadcopter control problem, results from which are briefly summarized here. The quadcopter is modelled as shown in \cite{bemporad2009hierarchical}. The system has $ x\in\mathbb{R}^{12} $ states and $ u \in \mathbb{R}^4 $ inputs. A nonlinear model predictive controller is designed to control the quadcopter to a desired target $ x_{target} $  described by an xyz-coordinate. The NMPc problem has a stage cost $ \ell(x,u) ) \|x - x_{target}\|_{Q}^2  + \|u\|_{R}^2 + \|\Delta u\|_{S}^2$. The goal is to train a control policy $ \pi(p_{i}; \theta) $ that takes as input the current states, control action, and the desired target, to compute the next action. That is, in this case $ p \in \mathbb{R}^{19} $. For this high dimensional system, the goal is to generate training data points using the proposed data augmentation scheme, which can be used to train a function approximator $ \pi: \mathbb{R}^{19} \rightarrow \mathbb{R}^4 $. 

$ N_s= 20  $ random samples are generated, where the MPC problem is solved exactly. We further randomly sample m=500 data points with ±20\% variation around each exact MPC sample, where the data is augmented from the MPC solution using the proposed method. By doing so, we generated a total of 10020 samples. For the predictor-corrector data augmentation, the corrector steps were used until the optimality residual was less than a tolerance of $ \epsilon_{tol}= 1e-2 $. 

Fig.~\ref{Fig:QuadError} shows the approximation error $ \|\hat{u} (p_i )-u^*(p_i)\| $ when using predictor-only data augmentation \cite{DK2021DataAug} (shown in blue), and the predictor-corrector data augmentation (shown in red). Their corresponding computation times are shown in Fig.~\ref{Fig:QuadCPU}, compared against generating all the samples exactly using the full NLP (yellow). To generate all the 10020 samples exactly required a total CPU time of 850.3s, whereas predictor-corrector data augmentation took 100.3s, and predictor-only data augmentation took 2.3s. This clearly shows the trade-off between the computation cost and the accuracy of the augmented samples.

\begin{figure*}
	\begin{subfigure}{0.5\textwidth}
			\centering
		\includegraphics[width=\linewidth]{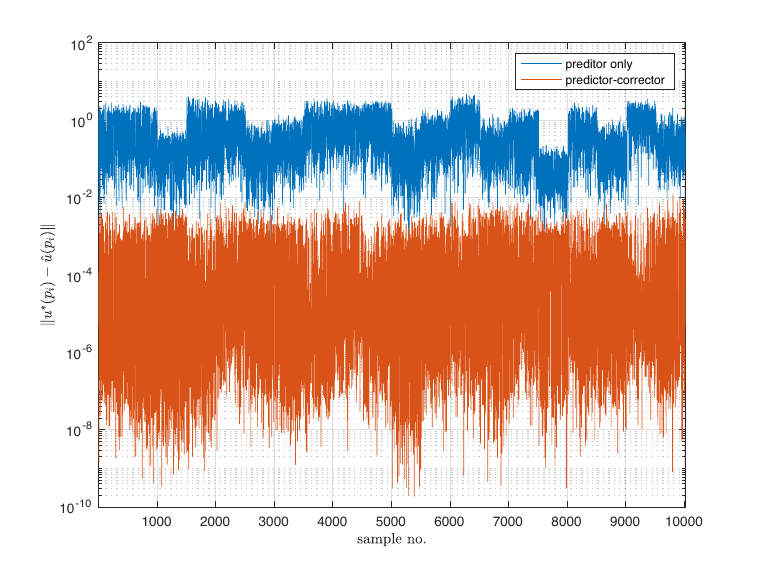}
		\caption{}\label{Fig:QuadError}
	\end{subfigure}
	\begin{subfigure}{0.5\textwidth}
	\centering
	\includegraphics[width=\linewidth]{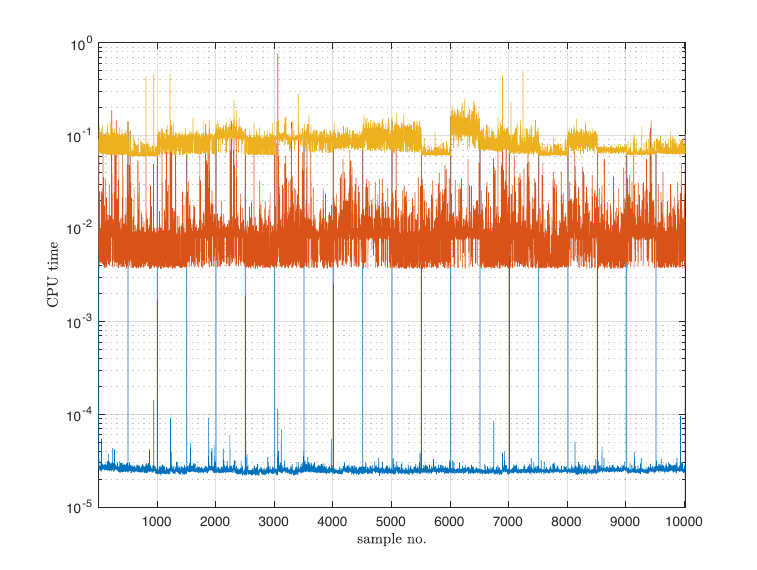}
	\caption{}\label{Fig:QuadCPU}
\end{subfigure}
\caption{Example 2 -  (a) Approximation error $ \|\hat{u} (p_i )-u^*(p_i)\| $ when using predictor-only data augmentation (blue), and predictor-corrector data augmentation (red). (b) CPU times for generating the full NLP data set (yellow), predictor-only augmented data set (blue), and predictor-corrector augmented data set (red). These two plots clearly show the trade-off between accuracy and the computation cost when generating training data to learn the MPC policy.}
\end{figure*}

\fi

\ifpreprint
\section{Discussion}
In this section, we provide some discussions on the placing our work on data augmentation in the context of other related literature. 
\subsection{Data augmentation and sampling techniques}
As mentioned in the introduction, generating large data sets that can be used for training and testing is critical for scaling MPC policy approximations. This challenge, if at all considered in the current works, is mainly addressed via efficient sampling techniques. 
This warrants a discussion on how our proposed data augmentation scheme fits in the context of  the different sampling strategies that have been proposed in the MPC literature. 

\paragraph{Adaptive sampling using resolution criteria} A recursive sampling algorithm was proposed in \cite{nubert2019mscthesis}, where new   sample locations are chosen with higher resolution as long as neighboring values are still differing.
At each sample location, the optimization problem is solved exactly to get the corresponding optimal action.  One can use the same algorithm as in \cite{nubert2019mscthesis} to decide where to sample the state-space, and if the newly generated sample locations are within a  neighborhood of a previously computed sample satisfying the assumptions of Theorem~\ref{thm:DataAugment}, then one can use the proposed data augmentation instead of solving the exact NLP at the newly sampled states. 

\paragraph{Sampling based on geometric random walks} The authors in \cite{chen2022large} recently proposed an algorithm to efficiently generate large data sets based on geometric random walks instead of independent sampling, where starting from a feasible state, small steps are taken iteratively along a random line until each step is feasible. Such an algorithm  naturally lends itself to our proposed data augmentation scheme, since this algorithm is based on the fact that  new samples are generated
based on previous successful samples. At every new step taken along the random walk, the corresponding optimal action can can be approximated using the previous sample, instead of solving it exactly. 

\paragraph{Control-oriented sampling} Based on the premise that  it is unlikely that all states from the feasible region are
equally likely to be observed during closed-loop operation, control oriented sampling strategies have been proposed in \cite{DK2021offsetENMPC}. In this approach, starting from a feasible initial state, the consecutive samples are selected based on noisy closed-loop simulations, such that the  samples are selected from a tube of closed loop trajectories likely to be observed in operation. Again, such an algorithm naturally lends itself to our proposed data augmentation scheme, since the  new samples are generated
based on previous successful samples. 

In general, the proposed data augmentation technique is complementary to the different sampling strategies, in the sense that, once the sample location is chosen using any sampling technique, instead of solving the optimization problem exactly at the chosen sample location, the proposed data augmentation technique can be used to efficiently approximate the corresponding optimal action with a desired accuracy at the given sample location. 
Given that the main motivation for developing  different sampling strategies is to scale MPC policy approximation to larger systems, combining these developments  with our proposed data augmentation scheme would enhance this goal. As such, the data augmentation scheme proposed in this paper complements the  developments on sampling strategies for offline training. 

\subsection{Data augmentation for a broader class of learning from demonstration problems}
The cost of generating sufficiently rich training data set  is also an open challenge in  a broad class of learning from demonstrations, such as imitation learning with interactive expert, and inverse optimal control. In this subsection, we briefly describe how the proposed data augmentation method can be used with a broad class of learning from demonstrations.

\paragraph{Imitation learning with interactive expert}
Imitation learning with interactive expert  is similar to the direct policy approximation approach described above, but the main difference is that we have access to an interactive expert in-the-loop, who we can query in order to improve our policy online \cite{argall2008learning,chernova2009interactive,ross2011reduction,jauhri2020interactive}.   For example, consider a system where we have an approximate policy running online, but we have an MPC controller running in the background that provides feedback on the optimal action for the state visited by the system. 

That is, we first train an initial policy $ \pi_k$ using the labeled data set  $ \mathcal{D}_{k}:= \{(x_{i},u_{i}^*)\}$. The current policy $ \pi_k $ is then rolled out to collect the trajectory  $ \{x_{t}, \pi_k(x_{t})\}_{t=1}^T $. For each trajectory, an interactive expert also provides the optimal policy for the states visited by the current roll-out (called interactive feedback), which are  aggregated such that $ \mathcal{D}_{k+1} = \mathcal{D}_{k} \cup \{(x_{t},u^*_{t})\}_{t=1}^T $ for $ k = 1, 2, \dots $, and the policy function  is updated on-the-fly by learning from  the aggregated data set $ \mathcal{D}_{k+1}$. 

In the case of an interactive expert that provides feedback on the rolled out policy (i.e.what the expert would have done in the states visited by the current policy), the proposed data-augmentation framework can be used to augment additional  samples around each state where the expert provides feedback. The \enquote{augmented} interactive feedback can then be used to improve the policy online. By doing so, we posit that this would enable us to generalize the feedback provided by the interactive expert to other nearby states not visited by the current policy (i.e. using the expert feedback, infer what the expert would have done in  other nearby states not visited by the current policy).  

The use of sensitivity-based data augmentation  for learning to control an inverted pendulum using imitation learning with interactive expert can be found in Appendix A.

\paragraph{Inverse Optimal Control}
An alternative approach  to direct policy approximation is to use the noisy observations of the optimal policy  to compute the cost function of a simpler optimization problem that is consistent with the expert demonstrations $ \mathcal{D}:= \{(x_{i},u_{i}^*)\}_{i=1}^N $. That is, we assume that $ \mathcal{D}:= \{(x_{i},u_{i}^*)\}_{i=1}^N $ comes from a simpler optimization problem, and we want to learn the cost function of this simpler optimization problem that is consistent with the demonstrations. By doing so, the policy is given by the simpler optimization problem, instead of the original complex optimization problem. For example, a framework for approximating  a complex MPC problem with a simple optimizing controller using inverse optimization was  shown in \cite{keshavarz2011imputing}.  
The proposed data augmentation scheme can also be used in such frameworks.

\paragraph{Value-function approximation}
Instead of approximating the policy function directly, one can obtain the optimal value function  $ V^*(x_{i}) $ from the expert queries, and approximate the optimal value function using the optimal state-value pairs $ \{(x_{i},V^*(x_{i}))\} $. The function approximator can then be used as the optimal cost-to-go function in simple myopic optimizing controllers in the context of approximate dynamic programming (ADP),  much like in  \cite{keshavarz2011imputing}(where instead of using a quadratic cost function,  any function approximator can be used). The proposed data augmentation framework can not only augment the optimal actions, but can also be used to augment the optimal state-value  pairs based on a few expert queries. To do this, (1) is solved for various state realizations $ \{x_{i}\} $, and the corresponding optimal value function $ V^*(x_{i}) $ is collected as the training data set. After each offline MPC solve, additional samples at $ x_{i} + \Delta x_{j} $ can be augmented using our proposed data augmentation scheme.  The  augmented data set $ \{(x_{i}, V^*(x_{i}))\} $ can then be used to fit a parametric function approximator that predicts the optimal value function $ V^*(x) $ for any given $ x $. This would be a better approach than learning the policy directly, since in this approach, we solve a much simpler and easier optimization problem, and by doing so, we can have feasibility guarantees by explicitly enforcing the constraints in the (simpler) optimization problem. 

\fi

\vspace{-0.2cm}
\section{Conclusion}
To summarize, this paper presented an improved data augmentation scheme for cheaply generating the training data samples for MPC policy approximation using only a fraction of the computation effort. We showed that by using additional corrector steps we can enforce a desired level of accuracy in the augmented samples (cf. Algorithm~1), which amounts to a few additional linear solves. By doing so the approximation error does not depend on the size of the neighborhood used for data augmentation, but on the user-defined optimality residual (cf. Theorem~ \ref{thm:MaxD2}), which allows us to augment samples from larger neighborhoods while without jeopardizing accuracy. This was demonstrated on a inverted pendulum control problem,  which clearly shows the trade-off between accuracy and computational effort. The proposed data augmentation approach can be used with any sampling strategy, as well as for a broad class of \textit{learning from demonstration} problems including imitation learning with interactive expert, value function approximation, and inverse optimal control. 

\bibliographystyle{IEEEtran}
\bibliography{L4DC}

\ifpreprint
\appendix
	
\subsection{Direct policy approximation with interactive expert}
The proposed data augmentation framework is applied  in the context of imitation learning with interactive expert on the same inverted pendulum example.  The objective here is to learn a policy function that would bring the pendulum from rest position ($ \omega = \dot{\omega} = 0 $) to its inverted position ($ \omega = 3.14, \dot{\omega} = 0 $). 

\paragraph{Task} Learn to invert a pendulum from rest. 
\paragraph{Initial policy} Linear policy $ u = -Kx $, with a$ K $  that is not able to accomplish the task.
\paragraph{Learning framework} 
We first start with no prior expert demonstrations, but with an initial  policy that is rolled out on the system. For this, we chose a linear policy $ u = -11\omega-7\dot\omega+35 $.  For each state visited in the rollout, the expert provides feedback on what the expert would have done in these states. Around each expert feedback, we further augment 25 samples using the proposed data augmentation framework. 

We first learn a policy function using only the expert feedback. To demonstrate the advantage of augmenting additional samples inferred from the expert feedback, we then learn a policy function based on the expert feedback and the augmented samples. For both these policies, we choose a generalized regression neural network that was trained using the  \texttt{newgrnn} function in \texttt{ MATLAB v2020b}.
The newly obtained approximate policy is then rolled out and this procedure is repeated for 25 times.  

Fig.~\ref{Fig:IEAll} left subplots shows the  states that are visited in each rollout, where the expert provides feedback (shown in red circles).  As it can be seen, even after 25 rollouts the number of demonstrations in the state-space is not much, which can affect the quality of the policy learned using only the expert feedback. Whereas, by augmenting additional samples around each expert feedback, we can cheaply generate several data samples covering a wider state-space than just what was visited in the rollouts. This is shown in black dots in Fig.~\ref{Fig:IEAll} right subplots. (From the 3D plot, one can already visually recognize the shape of the optimal solution manifold when we augment additional data samples). 

\begin{figure*}
	\centering
	\includegraphics[width=\linewidth]{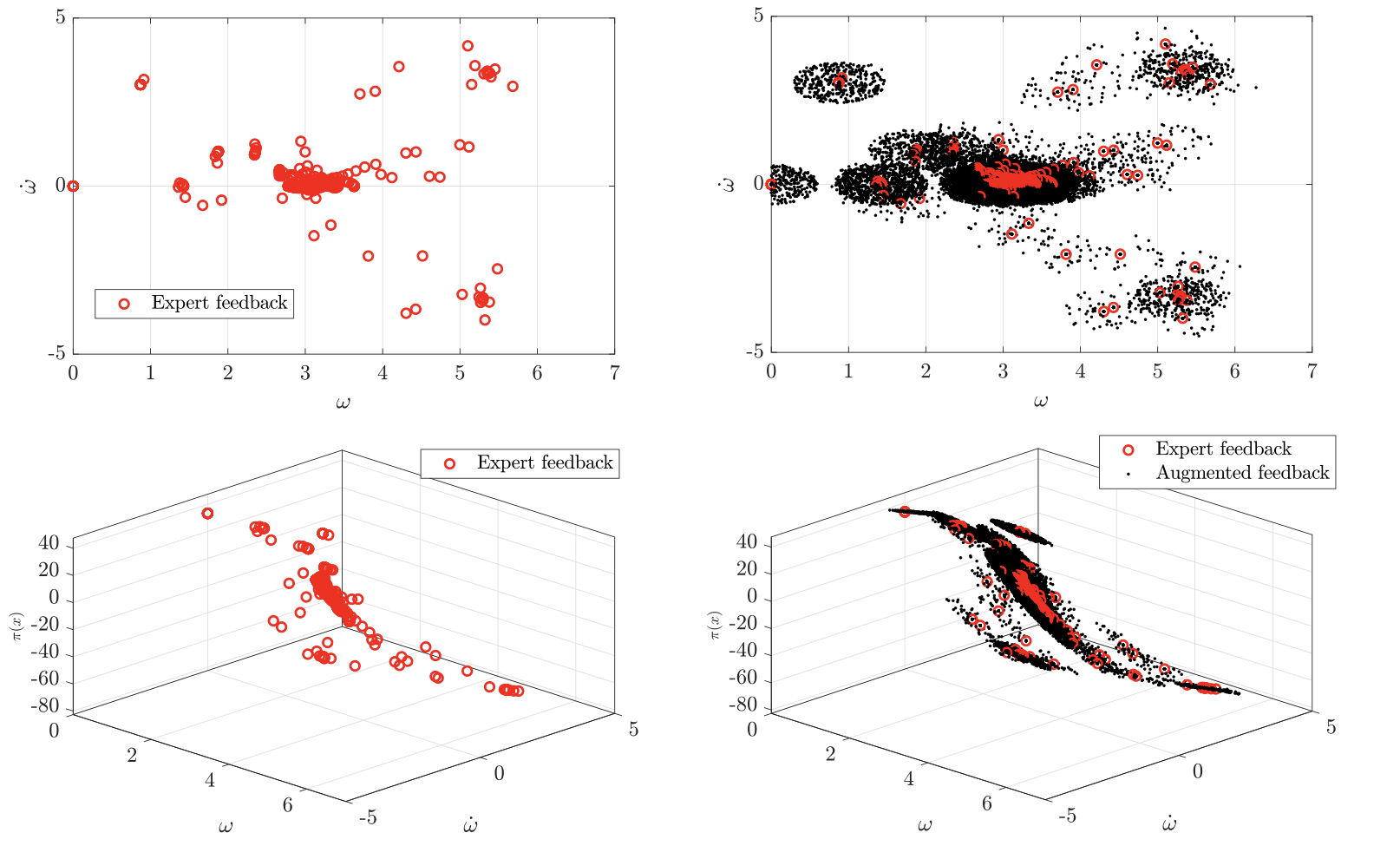}
	\caption{Closed loop performance over 15 rollouts when using only the expert feedback (shown in red), and using botht he expert and augmented feedback (shown in black). The bold lines shows the performance at the 15$ ^{th} $ rollout.} \label{Fig:IEAll}
\end{figure*}

\begin{figure*}
	\centering
	\includegraphics[width=\linewidth]{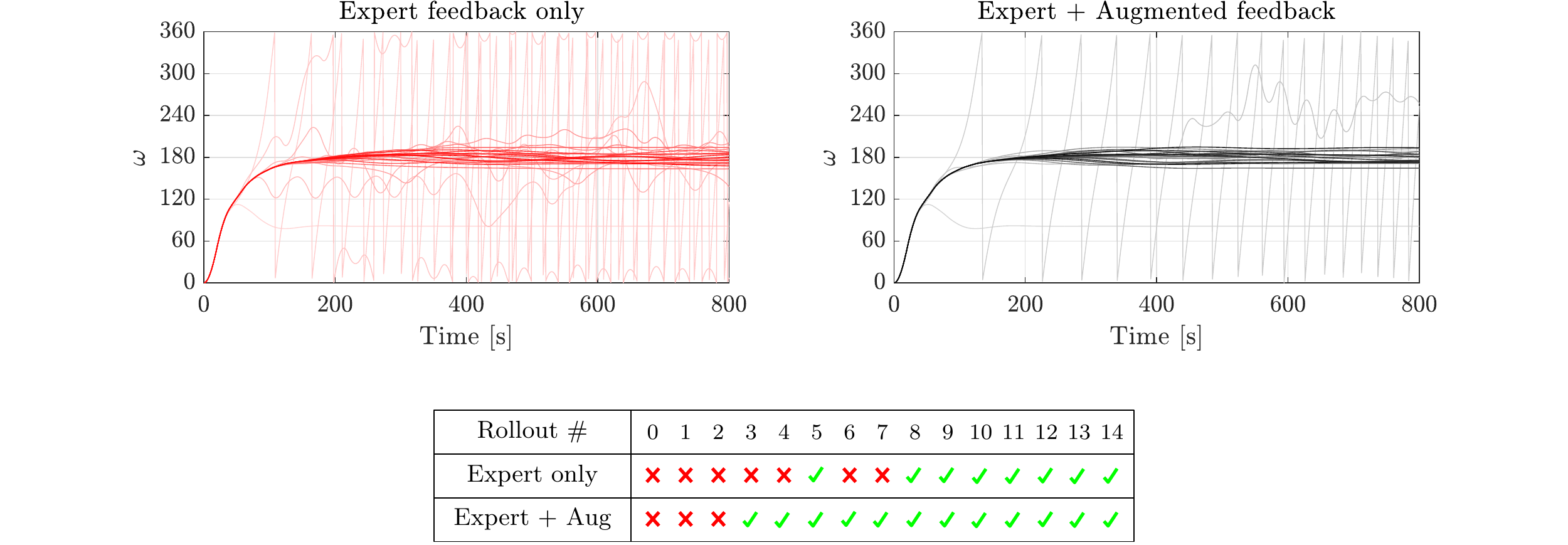}
	\caption{Closed-loop performance of all the rollouts using only expert feedback (left subplot) and using expert and augmented feedback (right subplot). Table shows whether the task in a particular rollout was successful or not.  } \label{Fig:InteractiveFeedbackAll}
\end{figure*}

The performance of the  policy learned using only the expert feedback   or using the expert and the augmented feedback for all the 25 rollouts are shown in Fig.~\ref{Fig:InteractiveFeedbackAll}.  Here it can be seen that by augmenting additional samples, we are able to better learn the policy, and hence learn to control the inverted pendulum better than if we only use the expert feedback. This is because, after each rollout, we not only learn from the expert feedback, but also infer what the expert would have done in other states not visited in the current rollout.  
\fi
\end{document}